%% file: spec.tex
\documentclass{article}
\input{commands.tex}

\usepackage[affil-it]{authblk}
\author{Nilin Abrahamsen}
\affil{\small{Department of Mathematics\\Massachusetts Institute of Technology\\Cambridge, MA, USA}}

\title{Short proof of a spectral Chernoff bound\\for local Hamiltonians}

\begin{document}

\maketitle

\input{abstract.tex} 
\input{specintro.tex}

\input{proof.tex}

\bibliographystyle{alpha}
\bibliography{bib.bib}
\end{document}

%% file: commands.tex
\usepackage{amsthm}
\usepackage[utf8]{inputenc}
\usepackage{thmtools}
\usepackage{thm-restate}
\usepackage{amssymb}
\usepackage{amsmath}
\usepackage{enumitem}
\usepackage{hyperref}
\usepackage{xcolor}
\usepackage[title]{appendix}
\usepackage{accents}

\hypersetup{colorlinks,linkcolor={red!70!black},citecolor={blue!70!black}}

\newtheorem{lemma}{Lemma}
\newtheorem{definition}{Definition}

\newtheorem{corollary}{Corollary}

\newtheorem{prob}{Problem}
\newtheorem{question}{Question}

\newtheorem{proposition}{Proposition}

\newcommand\DM{{\accentset{\infty}{g}}}
\newcommand\DA{{\bar g}}
\newcommand\D{{g}}
\newcommand\G{\mathcal G}
\newcommand\V{\mathcal V}
\newcommand\E{\mathcal E}
\newcommand\EE{\mathbb E}
\newcommand\RR{\mathbb R}
\newcommand\CC{\mathbb C}
\newcommand\HS{\mathcal H}
\newcommand\ESD{\sigma}
\newcommand\ket[1]{|#1\rangle}
\newcommand\bra[1]{\langle#1|}
\newcommand\opn[1]{\operatorname{#1}}
\newcommand\tr{\opn{tr}}
\newcommand\one{1\kern-.8ex\opn{I}}

\newcommand\rank{\opn{rank}}
\newcommand\Ecdf{\mathbb{F}}


%% file: abstract.tex
\begin{abstract}
	We give a simple proof of a Chernoff bound for the spectrum of a $k$-local Hamiltonian based on Weyl's inequalities. The complexity of estimating the spectrum's $\epsilon(n)$-th quantile up to constant relative error thus exhibits the following dichotomy: For $\epsilon(n)=d^{-n}$ the problem is NP-hard and maybe even QMA-hard, yet there exists constant $a>1$ such that the problem is trivial for $\epsilon(n)=a^{-n}$. We note that a related Chernoff bound due to Kuwahara and Saito (Ann. Phys. '20) for a generalized problem is also sufficient to establish such a dichotomy, its proof relying on a careful analysis of the \emph{cluster expansion}.  
\end{abstract}

%% file: specintro.tex

\section{Introduction}

A fundamental problem in the intersection of quantum physics and computer science is that of computing the \emph{energy levels} of a system of $n$ interacting particles. These are the eigenvalues of the \emph{local Hamiltonian} $H$, a conjugate-symmetric (\emph{Hermitian}) linear operator acting on the tensor product $\HS\simeq(\CC^d)^{\otimes n}$. 
The locality property means that $H$ is a sum of terms $H_\eta\otimes I$ where $H_\eta$ is an operator on $k=O(1)$ tensor factors and $I$ is the identity on the remaining factors. The locality structure gives rise to a hypergraph $\G=(\V,\E)$ with $|\V|=n$ and with the $H_\eta$ indexed by $m$ hyperedges $\eta\in\E$. 
Standard diagonalization procedures to compute the energy levels would take exponential time due to the dimension of the tensor product space.

The most famous problem in this category focuses on computing the \emph{lowest} eigenvalue, the \emph{ground state energy}. This generalizes the problem MAX-CSP of computing the optimal value of a constraint satisfaction problem, but now the ``variable assignments'' are vectors with exponentially many parameters. Computing the lowest eigenvalue up to a certain inverse polynomial accuracy in known to be complete for QMA \cite{kempe_complexity_2006}, a quantum analogue of NP. 
A major open problem is the \emph{quantum PCP-conjecture} \cite{aharonov_guest_2013} which posits that it is QMA-hard to even \emph{approximate} the ground state energy of the Hamiltonian $H=\sum_{\eta\in\E}H_\eta$ up to constant \emph{relative} error $\gamma m$. Here, $\|H_\eta\|\le1$ for each of the $m$ interactions $\eta\in\E$, and $\gamma$ is a small constant.

A number of approximation algorithms for local Hamiltonians have been put forth \cite{anshu_beyond_2020,hallgren_approximation_2020,brandao_product-state_2016,bravyi_approximation_2019}. Successful approximation algorithms imply no-go theorems for the quantum PCP conjecture, imposing restrictions on the possible hard instances that would make the conjecture true. Indeed it suffices to place the approximation problem in NP which is thought to be strictly smaller than QMA.

A related classic question in physics asks about the \emph{distribution} or \emph{density} of energy levels. \cite{jensen_quantum_2020} recently proposed \emph{quantum} algorithms for this question, which can be phrased in terms of computing the number of eigenvalues in a given interval. The complexity of the spectral density for \emph{local} Hamiltonians was studied in \cite{brown_computational_2011}, where it was shown that computing the number of eigenvalues in an interval of inverse polynomial length is no harder than $\#$P, subject to an inverse-polynomial gap around the interval. \cite{harrow_classical_2020} gives classical algorithms to compute \emph{partition functions} of local Hamiltonians, which similarly characterizes the aggregate behavior of many eigenvalues.

Combining the ideas of approximation algorithms and spectral density estimation raises the question: \emph{Can we construct an efficient {approximation} algorithm for the spectral distribution of a local Hamiltonian?} The \emph{empirical spectral distribution} (ESD) of $H$ is the probability distribution $\ESD_H=\frac1{\dim\HS}\sum_{i}\delta_{\lambda_i}$ where $\delta_{\lambda_i}$ is the point probability measure at the $i$\textsuperscript{th} eigenvalue $\lambda_i$ (with multiplicity, in non-decreasing order). By \emph{approximation} we mean that we allow errors along the horizontal (eigenvalue) axis when viewing the distribution as a histogram. 

We compare with a result from high-dimensional statistics \cite{kong_spectrum_2017}: Given i.i.d. samples of a $D$-dimensional random vector $Y$, estimate the spectrum of $Y$'s \emph{covariance} matrix. \cite{kong_spectrum_2017} showed that the spectrum of the covariance matrix can be approximated using a number of samples \emph{sublinear} in the dimension, and hence with much fewer samples than would be needed to approximate the covariance matrix itself (in particular the sample covariance matrix is low-rank so most of its eigenvalues are $0$).
The quality of approximation in \cite{kong_spectrum_2017} is evaluated in terms of the \emph{earth-mover's} distance (also called Wasserstein-1 distance, written $W^1$), which allows but penalizes errors along the horizontal eigenvalue axis of the histogram.
 The spectrum estimation is achieved by estimating the low-degree \emph{moments} of the spectrum.

We now note that in the setting of local Hamiltonians we are also able to compute the constant-degree moments efficiently. For operators acting on a vector space $\HS$ introduce the \emph{normalized trace} $\bar\tr=\frac1{\dim\HS}\tr$. Consider the rescaled Hamiltonian $h=\frac1mH$ and its empirical spectral distrbution $\tilde\ESD_h$. The $r$\textsuperscript{th} moment of $h$'s spectrum can be written as:
\begin{equation}\label{eq:mnteq}\int t^r d\tilde\ESD_h(t)=\bar\tr(h^r)=\EE_{\eta_1}\cdots\EE_{\eta_r}\bar\tr(H_{\eta_1}\cdots H_{\eta_r}),\end{equation}

where the $\eta_i$ are sampled i.i.d. from the uniform distribution on interactions $\eta\in\E$.
As is convention we use $H_\eta$ as a shorthand for $H_\eta\otimes I$. Note that, unlike the standard trace, $\bar\tr$ is unchanged when tensoring with the identity. This follows easily from noticing that $\bar\tr(H)=\EE\bra\psi H\ket\psi$,\footnote{The row vector $\bra\psi$ is the dual, or conjugate transpose, of column vector $\ket\psi$.} where $\ket\psi$ is chosen uniformly at random from an orthonormal basis.

Since $H_{\eta_1}\cdots H_{\eta_r}$ acts on the set $\eta_1\cup\cdots\cup\eta_r$ of at most $rk$ qudits, each term $\bar\tr(H_{\eta_1}\cdots H_{\eta_r})$ can be computed in time $O(rd^{2.38rk})$.  
\eqref{eq:mnteq} immediately yields an algorithm to approximate the spectrum of $H$ up to small relative error $\gamma$ in time independent of $m$. Indeed, \cite{kong_spectrum_2017} proposition 1 implies that for a distribution of bounded support (the spectrum of $h$ in this case), knowing the first $r=\lfloor C\gamma^{-1}\rfloor$ moments gives an $\gamma$-approximation in $W^1$ distance. Moreover, it suffices to approximate each moment up to an error exponentially decreasing in $r$. So it suffices to sample $2^{O(r)}$ terms in \eqref{eq:mnteq} and compute each in time $d^{O(rk)}=d^{O(k/\gamma)}$ for a total time complexity of $d^{O(k/\gamma)}$. 

The questions remains: does the output of the above moment-based algorithm give us nontrivial information about the spectrum of $H$, or will it instead be an expression of a universal property of a local Hamiltonian's spectrum which could be known without running the algorithm? It turns out that the latter is the case, as shown by the following simple computation: Let $\mu_\eta=\bar\tr(H_\eta)$. Applying the $r=2$ case of \eqref{eq:mnteq} to the centered interaction terms $H_\eta-\mu_\eta$ we get
\begin{equation}\label{eq:var}\opn{Var}\tilde\ESD_h=\EE_\eta\EE_{\eta'}\bar\tr\big((H_{\eta}-\mu_{\eta})(H_{\eta'}-\mu_{\eta'})\big)\le 4\mathbb{P}(\eta\cap\eta'\neq\emptyset),\end{equation}
since the terms evaluate to $\bar\tr(H_{\eta}-\mu_{\eta})\bar\tr(H_{\eta'}-\mu_{\eta'})=0\cdot0$ when the interactions do not overlap. Assume for simplicity that every vertex is involved in $\D$ interactions $\eta$ and every interaction involves $k$ qudits. Then any $\eta$ overlaps with at most $k\D$ other hyperedges in the interaction hypergraph, fixing $\eta$ we have $\mathbb{P}(\eta'\cap\eta\neq\emptyset|\eta)\le k\D/m$. Writing $mk=\D n$, \eqref{eq:var} implies,
\begin{equation}\label{eq:vartwo}\opn{Var}\tilde\ESD_h\le 4k^2/n.\end{equation}
Note that \eqref{eq:vartwo} does not depend on the vertex degree $\D$.  
It follows that the point measure at $\tilde\mu=\bar\tr(h)$ approximates $\tilde\ESD_h$ up to error $4k^2/n$ in $L^2$-distance, and $2k/\sqrt n$ in earth-mover's distance, by Cauchy-Schwartz or Jensen's inequality (The $W^1$ distance coincides with the $L^1$-distance since one distribution is a point). This makes the output of the aforementioned moment-based algorithm trivial for low-degree moments. Indeed, it would require degree $r\gg\sqrt n$ moments and time complexity $e^{\omega(\sqrt n)}$ to improve on the trivial estimate $\delta_{\tilde\mu}$ of $\tilde\ESD_h$.

We will not use the rescaled operator $h$ in the remainder of the paper.

\subsection{Spectrum estimation in terms of quantiles}
We relate the problem of spectrum estimation to the topic of ground states: Generalize the problem of approximating the ground state energy to that of approximating the $\epsilon(n)$\textsuperscript{th} quantile of $H$'s spectrum up to constant relative error $\gamma$.  
\begin{prob}
	\label{prob:approxproblem}
	Given a $k$-local Hamiltonian $H=\sum_{\eta\in\E}H_\eta$ with $\|H_\eta\|\le1$ and $m=|\E|$ encoded as the list $(H_\eta)_{\eta\in\E}$, output $\hat\lambda$ such that $|\hat\lambda-\lambda_i(\bar H)|\le \gamma m$ where $i=\lfloor\epsilon(n)d^n\rfloor$.
\end{prob}

We ask how the complexity of problem \ref{prob:approxproblem} depends on $\epsilon(n)$. By symmetry we may restrict attention to $\epsilon(n)\le1/2$. For $\epsilon(n)=d^{-n}$ problem \ref{prob:approxproblem} is conjectured to be QMA-hard according to the quantum PCP-conjecture.

On the other hand \eqref{eq:var} already implies the weak concentration inequality (Chebyshev's inequality) $\ESD_H([\mu-\gamma m,\mu+\gamma m]^\complement)\le \gamma^{-2}k^2/n$ where $\mu=\bar\tr(H)$ and $\complement$ denotes the complement, so for $\epsilon(n)\ge\gamma^{-2}k^2/n$ we may simply output $\mu$. 

\begin{question}\label{q:thequestion}{Can this concentration bound be strengthened to be exponentially decreasing in $n$, showing that problem \ref{prob:approxproblem} is easy for some exponentially decreasing $\epsilon(n)$}? 
\end{question}
It turns out that the answer is yes.
The technical contribution of this paper is a simple proof of such a Chernoff-type bound with exponentially decreasing tails. Our proof is based of Weyl's eigenvalue inequalities. Note that a similar bound was previously known from a careful analysis of \emph{cluster expansions} by Kuwahara and Saito (\cite{kuwahara_gaussian_2020} corollary 2). In fact their result is for the more general energy distribution $\ESD^\rho$ of certain states $\rho$ relative to $H$. Other previous works \cite{anshu_concentration_2016,kuwahara_connecting_2016} give bounds for $\ESD^\rho$ using a moment-based approach, but their bounds are restricted to a \emph{short-range interacting} setting and are therefore less similar to ours (See section \ref{comparison}).
A different spectral concentration inequality by \cite{montanaro_applications_2012} does not include the dependence on system size which is of central interest here.

\section{Statement of the Chernoff bound}

We briefly recall out notation. $\G=(\V,\E)$ is a hypergraph with $|\V|=n$ vertices and $|\E|=m$ hyperedges, each incident to $k$ vertices. $\V$ indexes the set of qudits, each isomorphic to $\CC^d$.
A local Hamiltonian is the sum $H=\sum_{\eta\in\E}H_\eta$ where each local interaction $H_\eta$ acts on the qudits in $\eta$ and $\|H_\eta\|\le1$. Let $\lambda_1( H)\le\lambda_2( H)\le\ldots$ be the ordered eigenvalues of $ H$.
The empirical spectral distribution (ESD) of Hamiltonian $H$ is the probability measure $\ESD_H=d^{-n}\sum_i\delta_{\lambda_i}(H)$ which assigns mass $d^{-n}$ to each of its $d^n$ eigenvalues counted with multiplicity. 

The \emph{vertex degree} $\D_v=|\{\eta\in\E|v\in\eta\}|$ of qudit $v\in\V$ is the number of interactions involving qudit $v$. Let $\DM=\max_{v\in\V}\D_v$ be the maximum degree over all qudits and let $\DA=\frac1n\sum_{v\in\V}\D_v$ be the average degree. $\DA$ and $\DM$ may be unbounded.
	
\begin{proposition}\label{prop:theprop}
	Let $H$ be a local Hamiltonian on a $k$-uniform hypergraph with maximum degree $\DM$ and average degree $\DA$. Let $\ESD$ be the ESD of $H$. Then,
	\begin{equation}\label{eq:yesfloor}\ESD\big([-m,{\mu}-\gamma m]\big)\le k\DM \exp\Big(-\frac{\gamma^2}2\Big\lfloor\frac{n}{k^2(\DM/\DA)}\Big\rfloor\Big),\end{equation}
	where ${\mu}=d^{-n}\tr H$. 	
	The same bound holds for $\ESD\big([{\mu}+\gamma m,m]\big)$.
\end{proposition}

We thus obtain an exponential concentration bound for arbitrarily large vertex degrees, as long as the average and maximum degree are of the same order. {The multiplicative prefactor $\DM\lesssim\DA=km/n\le \frac kn\binom{n}{k}=n^{O(1)}$ can be absorbed}. 
\begin{corollary}\label{cor:thecor}
Let $H$ be a local Hamiltonian on a $k$-uniform and $\D$-regular hypergraph. Then,
\[\ESD\big([-m,{\mu}-\gamma m]\big)\le k\D \exp\Big(-\frac{\gamma^2}2\lfloor{n}/{k^2}\rfloor\Big),\]
where $\mu=d^{-n}\tr H$. The same bound holds for $\ESD\big([{\mu}+\gamma m,m]\big)$.
\end{corollary}

Proposition \ref{prop:theprop} answers question \ref{q:thequestion} of the previous section about the complexity of problem \ref{prob:approxproblem}, assuming that the maximum degree and average degree are of the same order. It implies that for any $\gamma=\Omega(1)$ there exists a constant $a>1$ such that problem \ref{prob:approxproblem} is trivial for $\epsilon(n)\ge a^{-n}$. The approximation is simply $\hat\lambda:=\sum_{\eta\in\E}\bar\tr(H_\eta)$.

\subsubsection{The importance of unbounded degree}
For the case $k=2$, \cite{brandao_product-state_2016} show that approximating the energy of $\D$-regular \emph{graphs} of high degree $\D$ is in NP. Moreover, \cite{bravyi_quantum_2008} constructed gadgets to reduce $k$-local interactions to 2-local interactions. These two facts may at first appear to imply that high-degree hypergraphs do not make for hard instances for the approximate ground state problem. It turns out (\cite{harrow_personal_nodate}) that this argument is not valid. Consider for example an input with $n$ \emph{qubits} and $m$ interactions, each $3$-local. The gadgets of \cite{bravyi_quantum_2008} produce a $2$-local Hamiltonian on a graph $\widetilde{\mathcal G}$ with $\tilde n=n+m$ vertices and $\tilde m\le 6m$ edges, with $m$ \emph{mediator} qubits added. Now the numbers of vertices and edges are of the same order $\tilde n=\Omega(\tilde m)$, hence the averaging argument of \cite{brandao_product-state_2016} does not imply a bound $o(1)$. 

As an illustration (\cite{harrow_personal_nodate}), applying \cite{brandao_product-state_2016} theorem 9 (the non-regular case) yields a relative error of order $(\|A\|_{F}^2\|\pi\|_2^2)^{1/8}$ where $\pi$ is a probability distribution on the enlarged vertex set and $\|A\|_F^2$ is the \emph{harmonic mean\footnote{The notation is interpreted as follows: $\|A\|_F^2$ is the squared Frobenius/Hilbert-Schmidt norm of $A$, the adjacency matrix of $\widetilde\G$ rescaled to have column sums equal to $1$.} of the degrees} in $\widetilde{\mathcal G}$, so the contribution from the degree-$3$ mediator qubits yields $\|A\|_F\ge\frac13 m$. Since $\|\pi\|_2^2\ge1/\tilde n=\Theta(1/m)$, the relative error bound $(\|A\|_{F}^2\|\pi\|_2^2)^{1/8}$ does not converge to $0$. 

In conclusion it is not known that approximating eigenvalues for high-degree $k$-local Hamiltonians is in NP, so hypergraphs with high vertex degree $\D\gg1$ are an important setting to study for approximation theory and the quantum PCP conjecture.

\subsection{Comparison with bounds in the literature}\label{comparison}
The analysis of the spectrum of $H$ is a special case of a problem studied in the literature seeking the distribution of an observable $H$ \emph{in a state $\rho$} \cite{anshu_concentration_2016,kuwahara_connecting_2016,kuwahara_gaussian_2020}. The state $\rho$ is subject to certain assumptions (say, product structure \cite{anshu_concentration_2016} or being a Gibbs state for a local Hamiltonian \cite{kuwahara_gaussian_2020}), and the distribution in question is $\ESD^{\rho}=\sum_{\lambda\in\opn{spec} H}\tr(\rho\Pi_\lambda)\delta_\lambda$, where $\Pi_\lambda$ are the spectral projections of the Hamiltonian $H$ and $\delta_\lambda$ the point probability measure at $\lambda$.\footnote{Here, the sum is over the spectrum of $H$ as a \emph{set}. Multiplicity of eigenvalues is included through the rank of $\Pi_\lambda$.} We call $\ESD^\rho$ the \emph{directional} energy distribution in direction $\rho$. This specializes to the spectral distribution of $H$ when $\rho$ is maximally mixed.

Consider a $k$-uniform, $\D$-regular interaction hypergraph but let us allow unbounded degree $\D\to\infty$. 
\cite{kuwahara_connecting_2016}, corollary 8 bounds $\ESD^\rho\big([-m,{\mu}-x]\big)\leq e^{-\tilde\Omega(x^2/n)}$ where the implicit constant depends on $k$ and, notably, on $\D$. Thus we must take $\D=O(1)$ which imposes that $m=\Theta(n)$. Substituting $x=\gamma m$ yields $\ESD^\rho\big([-m,\mu-\gamma m]\big)\leq e^{-\tilde\Omega(n\gamma^2)}$ as in proposition \ref{prop:theprop} for $k,\D=\Theta(1)$, but one does not get a bound when $\D\to\infty$.  
\cite{anshu_concentration_2016} theorem 1.2 gives a bound with an explicit dependence on $\D$,
\begin{equation}\label{eq:anshu16}\ESD^\rho\big([-m,{\mu}-\gamma m]\big)=e^{-\Omega(\frac{m\gamma^2}{k^2\D^2})}=\exp\Big({-\Omega\big(\frac{n\gamma^2}{k^3\D}\big)}\Big).\end{equation}
\eqref{eq:anshu16} does not obtain the exponential decay in $n$ as in proposition \ref{prop:theprop} unless $\D=O(1)$. For example, for the case of the complete graph the bounds of \cite{kuwahara_connecting_2016} and \cite{anshu_concentration_2016} do not show any concentration, whereas proposition \ref{prop:theprop} decreases exponentially with $n$. Physically, these limitations correspond to saying that the results of \cite{anshu_concentration_2016,kuwahara_connecting_2016} are for \emph{short-range} interacting systems.

The bound which is most similar to ours is found in \cite{kuwahara_gaussian_2020} and uses a delicate analysis of the \emph{cluster expansions} to obtain a bound on the energy distribution in \emph{long-range} interacting systems. \cite{kuwahara_gaussian_2020} corollary 2 states the bound for the spectral distribution (the same setting as ours):
\begin{equation}\label{eq:ks19}\ESD\big([-m,{\mu}-\underbrace{\gamma m}_x]\big)\le\exp\Big(-\frac{x^2}{(16e^3\D k)\D n}\Big)=\exp\Big(-\frac1{16e^3}\frac{\gamma^2n}{k^3}\Big),\end{equation}
where we have substituted $x=\gamma m=\gamma\cdot\D n/k$.
In corollary \ref{cor:thecor} we obtain an exponent of order $\frac12(\gamma/k)^2n$, improving over \eqref{eq:ks19} by a factor $160k=\Theta(k)$ in the exponent for our problem setting (since $160=\lfloor8e^3\rfloor$). 
For local Hamiltonians one has $k=O(1)$ so our result shrinks the \emph{base} of the exponential decay by a factor $e^{160k}\ge e^{320}$.
We stress that \cite{kuwahara_gaussian_2020} is able to analyze the more general \emph{directional} energy distribution $\ESD^\rho$. Furthermore the ``degree'' of a vertex $v$ is defined in a more flexible way in \cite{kuwahara_gaussian_2020} as a bound on $\sum_{\eta:v\in\eta}\|H_\eta\|$, and fewer-particle interactions $|\eta|< k$ are allowed. 

\subsubsection{Interpretation as typical directional energy distribution}
Estimating the directional energy distribution $\ESD^{\rho}$ is a more general problem than estimating the spectral distribution $\ESD$ of $H$. On the other hand, propositon \ref{prop:theprop} implies a partial converse (a similar connection was observed in \cite{montanaro_applications_2012}): Given any orthonormal basis $\{\ket{\psi_i}\}$ we have that the energy distribution in direction $\ket{\psi_i}$ satisfies an exponential concentration bound in all but an exponentially small proportion of the directions $\ket{\psi_i}$. Indeed, the spectral density bounded in proposition \ref{prop:theprop} can be written as $\ESD=\EE\ESD^{\ket{\psi_i}} $ where the expectation is over a uniformly chosen member of the basis. Letting $\epsilon$ be twice the RHS of proposition \ref{eq:yesfloor} and writing $\complement=\{x:|x-\mu|\ge \gamma m\}$ we get by Markov's inequality:
\[\sqrt\epsilon\cdot\mathbb P(\ESD^{\ket{\psi_i}}(\complement)>\sqrt\epsilon)\le\EE \ESD^{\ket{\psi_i}}(\complement)=\ESD(\complement)\le\epsilon\quad\Rightarrow\quad\mathbb P(\ESD^{\ket{\psi_i}}(\complement)>\delta)\le\sqrt\epsilon,\]
where $\sqrt{\epsilon}=Ce^{-\Omega(nt^2/k^2)}$. So for any $0<\gamma<1$ there exists a large set $S\subset\{1,\ldots,\dim\HS\}$ indexing basis vectors such that $|S|/\dim\HS\ge1-Ce^{-\Omega(n\gamma^2/k^2)}$, and such that the energy distribution concentrates,
\[\ESD^{\ket{\psi_i}}([\mu-\gamma m,\mu+\gamma m])\ge1-Ce^{-\Omega(n\gamma^2/k^2)},\]
in all directions $\ket{\psi_i}$, $i\in S$.

%% file: proof.tex
\section{Simple proof of spectral concentration}

We now turn to the proof of proposition \ref{prop:theprop}.  The idea of our proof is to use Weyl's eigenvalue inequalities \cite{weyl_asymptotische_1912} to combine multiple independent sets of interactions. A similar partitioning of interactions into independent sets has been done previously \cite{kuwahara_connecting_2016,anshu_concentration_2016} (see \cite{kuwahara_connecting_2016} lemma 2), but in these cases the sets were combined in a more elaborate way by analyzing the \emph{moments}, and the results do not yield our desired bounds in the long-range interacting case. 
\begin{definition}
	Given a Hermitian operator $H$ and $t\in\RR$, let $F(t)$ be the proportion of $H$'s eigenvalues in $(\infty,t]$. We call $F$ the ESD-CDF (cumulative distribution function for the ESD) of $H$.
\end{definition}

\begin{lemma}[Weyl's inequalities]\label{lem:Weyl}
	Let $H=\sum_{c=1}^{r} H_c$.
	For each $c$ let $F_c$ be the ESD-CDF of $H_c$, and let $F(t)$ be the ESD-CDF of $H=\sum_cH_c$. Then for any $t_1,\ldots,t_r\in\mathbb R$
	\[F(t_1+\ldots+t_r)\le F_1(t_1)+\cdots+F_r(t_r).\]
\end{lemma}
\begin{proof}
	Let $E_c$ be the projection-valued spectral measure for $H_c$ and consider the projection-valued CDF $\Ecdf_c(t)=E((-\infty,t])$. Then $F_c(t)=\frac1{\dim\HS}\rank \Ecdf_c(t)$.  
	Consider any subspace $W$ of dimension $D+1$ where $D=\sum_c\rank\Ecdf_c(t_c)$. $W$ contains a unit vector $\ket\psi$ in $\bigcap_c\ker\Ecdf_c(t_c)$. Then $\bra\psi H_c\ket\psi>t_c$ for each $c$, so $\bra\psi H\ket\psi> \sum_c t_c$. So by the Courant-Fischer min-max theorem, $\lambda_{D+1}(H)>\sum_c t_c$. That is, at most $D=(\dim\HS)\sum_cF_c(t)$ eigenvalues are in $(-\infty,\sum_c t_c]$.
\end{proof}

	A (hyper)edge coloring of $\G=(\V,\E)$ is a partition $\E=\E_1\sqcup\ldots\sqcup\E_r$ such that for any color $c$ it holds that any two distinct $\eta_1,\eta_1\in\E_c$ are disjoint as subsets of $\V$.

\begin{lemma}\label{lem:combined}
	Let $\E=\E_1\sqcup\ldots\sqcup\E_r$ be a hyperedge coloring and let $m_c=|\E_c|$ be the number of hyperedges with the color $c$. Let ${\mu}=\bar\tr(H)$. Then the ESD-CDF $F$ of $H$ satisfies
\begin{equation}\label{eq:colsum}F({\mu}-\gamma m)\le\sum_{c=1}^r\exp(-m_c\gamma^2/2).\end{equation}
\end{lemma}
\begin{proof}
	Write $H^{(c)}=\sum_{\eta\in\E_c}H_\eta$ and $\mu_c=\bar\tr(H^{(c)})$. By lemma \ref{lem:Weyl} we have,
	\begin{equation}\textstyle{F({\mu}-\gamma m)=F\Big(\sum_{c=1}^r(\mu_c-\gamma m_c)\Big)\le \sum_{c=1}^r F_c({\mu}_c-\gamma m_c).\label{eq:applyweyl}}\end{equation}
	We fix $c$ and bound $F_c({\mu}_c-\gamma m_c)$: For each $\eta\in\E_c$ choose independently a uniformly random $\ket{\psi_\eta}$ from an eigenbasis for $H_\eta$, so that $\ket{\psi}=\bigotimes_{\eta\in\E_c}\ket{\psi_\eta}$ is uniform from an eigenbasis for $H^{(c)}$. The corresponding random eigenvalue $\lambda$ is distributed according to the ESD of $H^{(c)}$. $\lambda$ is a sum of $m_c$ random variables in the interval $[-1,1]$, so by Hoeffding's bound (\cite{hoeffding_probability_1963} inequality (2.3)),
	\[F_c({\mu}_c-\gamma m_c)=P\Big(\lambda\le\EE\lambda-m_c\gamma\Big)=\exp(-m_c\gamma^2/2).\]
\end{proof}
To finish our proof of proposition \ref{prop:theprop} it remains to determine the number of colors $r$ and the sizes $m_c$ of the independent sets in \eqref{eq:colsum}.

\begin{lemma}\label{lem:coloring}
	There exists an equitable coloring $\E=\E_1\sqcup\cdots\sqcup\E_r$ with $r=k\DM-k+1$ colors. Here equitable means that $|\E_c|\ge\lfloor m/r\rfloor$ for each $c=1,\ldots,r$.
\end{lemma}
\begin{proof}
	Construct the graph $G^*$ on vertex set $\E$ where two interactions $\eta\sim\eta'$ are connected iff some qudit is acted on by both $\eta$ and $\eta'$.
$G^*$ has degree at most $k(\DM-1)$. By the Hajnal-Szemeredi theorem \cite{kierstead_short_2008} there exists an equitable vertex coloring of $G^*$ with $r=k(\DM-1)+1$ colors. 
\end{proof}

\begin{proof}[Proof of proposition \ref{prop:theprop}]
	We use lemma \ref{lem:coloring} to pick an equitable coloring with $r\le k\DM$ colors. Since the coloring is equitable we have $m_c\ge\lfloor m/r\rfloor$. Apply lemma \ref{lem:combined} and note that each term in \eqref{eq:colsum} is bounded by $\exp(-\lfloor m/r\rfloor \gamma^2/2)$. Thus,
	\[F({\mu}-\gamma m)\le r\exp(-\lfloor m/r\rfloor\gamma^2/2 )\le k\DM\exp\Big(-\frac{\gamma^2}2\Big\lfloor\frac{m}{k\DM}\Big\rfloor\Big),\]
	where we have used that the middle expression is increasing in $r$. The result follows by noting that $n\DA =\sum_v\D _v=mk$ and substituting $m=(\DA /k)n$
\end{proof}

\section{Open problems}
We conclude with a few questions raised by this work.

\begin{itemize}[label={--}]
	\item Can the bound of proposition \ref{prop:theprop} be matched for the directional energy distribution $\ESD^\rho$ considered in \cite{kuwahara_gaussian_2020}?
	
	\item
	We obtain an exponential concentration bound when the maximum degree and the average degree in the interaction hypergraph are of the same order, possibly unbounded. Can this regularity-like condition be weakened, either by allowing a few vertices of atypically large degree, or by defining the degree using operator norms as in \cite{kuwahara_gaussian_2020}? 

\item In the introduction we sketched a na\"ive moment-based algorithm for approximating the spectrum of a local Hamiltonian and noted that because of spectral concentration its output would be trivial when using low-degree moments. Could this algorithm be improved using combinatorial insights, say, by using the cluster expansions? Alternatively, running the na\"ive algorithm up to $\sqrt n\ll M \ll n$ moments estimates the spectrum to greater precision than the trivial point estimate in sub-exponential time; would this give non-trivial information about the spectrum? Or does one intead find that the spectral distribution is always close to a Gaussian, as in the case of spin chains \cite{keating_spectra_2015}? 
\end{itemize}

\section{Acknowledgements}
The author thanks Aram Harrow for helpful comments.

%% file: spec.bbl
\begin{thebibliography}{JKKAG20}

\bibitem[AAV13]{aharonov_guest_2013}
Dorit Aharonov, Itai Arad, and Thomas Vidick.
\newblock Guest {Column}: {The} {Quantum} {PCP} {Conjecture}.
\newblock {\em SIGACT News}, 44(2):47--79, June 2013.

\bibitem[AGM20]{anshu_beyond_2020}
Anurag Anshu, David Gosset, and Karen Morenz.
\newblock Beyond {Product} {State} {Approximations} for a {Quantum} {Analogue}
  of {Max} {Cut}.
\newblock In {\em 15th {Conference} on the {Theory} of {Quantum} {Computation},
  {Communication} and {Cryptography} ({TQC} 2020)}. Schloss
  Dagstuhl-Leibniz-Zentrum für Informatik, 2020.

\bibitem[Ans16]{anshu_concentration_2016}
Anurag Anshu.
\newblock Concentration bounds for quantum states with finite correlation
  length on quantum spin lattice systems.
\newblock {\em New Journal of Physics}, 18(8):083011, 2016.

\bibitem[BDLT08]{bravyi_quantum_2008}
Sergey Bravyi, David~P DiVincenzo, Daniel Loss, and Barbara~M Terhal.
\newblock Quantum simulation of many-body {Hamiltonians} using perturbation
  theory with bounded-strength interactions.
\newblock {\em Physical review letters}, 101(7):070503, 2008.

\bibitem[BFS11]{brown_computational_2011}
Brielin Brown, Steven~T Flammia, and Norbert Schuch.
\newblock Computational difficulty of computing the density of states.
\newblock {\em Physical review letters}, 107(4):040501, 2011.

\bibitem[BGKT19]{bravyi_approximation_2019}
Sergey Bravyi, David Gosset, Robert Koenig, and Kristan Temme.
\newblock Approximation algorithms for quantum many-body problems.
\newblock {\em Journal of Mathematical Physics}, 60(3):032203, March 2019.
\newblock arXiv: 1808.01734.

\bibitem[BH16]{brandao_product-state_2016}
Fernando G. S.~L. Brandão and Aram~W. Harrow.
\newblock Product-state {Approximations} to {Quantum} {Ground} {States}.
\newblock {\em Communications in Mathematical Physics}, 342(1):47--80, February
  2016.
\newblock arXiv: 1310.0017.

\bibitem[Har]{harrow_personal_nodate}
Aram~W. Harrow.
\newblock personal communication.

\bibitem[HLP20]{hallgren_approximation_2020}
Sean Hallgren, Eunou Lee, and Ojas Parekh.
\newblock An approximation algorithm for the {MAX}-2-{Local} {Hamiltonian}
  problem.
\newblock In {\em Approximation, {Randomization}, and {Combinatorial}
  {Optimization}. {Algorithms} and {Techniques} ({APPROX}/{RANDOM} 2020)}.
  Schloss Dagstuhl-Leibniz-Zentrum für Informatik, 2020.

\bibitem[HMS20]{harrow_classical_2020}
Aram~W Harrow, Saeed Mehraban, and Mehdi Soleimanifar.
\newblock Classical algorithms, correlation decay, and complex zeros of
  partition functions of quantum many-body systems.
\newblock In {\em Proceedings of the 52nd {Annual} {ACM} {SIGACT} {Symposium}
  on {Theory} of {Computing}}, pages 378--386, 2020.

\bibitem[Hoe63]{hoeffding_probability_1963}
Wassily Hoeffding.
\newblock Probability {Inequalities} for {Sums} of {Bounded} {Random}
  {Variables}.
\newblock {\em Journal of the American Statistical Association},
  58(301):13--30, 1963.

\bibitem[JKKAG20]{jensen_quantum_2020}
Phillip W.~K. Jensen, Lasse~Bjørn Kristensen, Jakob~S. Kottmann, and Alán
  Aspuru-Guzik.
\newblock {\em Quantum {Computation} of {Eigenvalues} within {Target}
  {Intervals}}.
\newblock 2020.

\bibitem[KK08]{kierstead_short_2008}
Hal~A Kierstead and Alexandr~V Kostochka.
\newblock A short proof of the {Hajnal}–{Szemerédi} theorem on equitable
  colouring.
\newblock {\em Combinatorics, Probability and Computing}, 17(2):265--270, 2008.

\bibitem[KKR06]{kempe_complexity_2006}
Julia Kempe, Alexei Kitaev, and Oded Regev.
\newblock The complexity of the local {Hamiltonian} problem.
\newblock {\em SIAM Journal on Computing}, 35(5):1070--1097, 2006.

\bibitem[KLW15]{keating_spectra_2015}
J.~P. Keating, N.~Linden, and H.~J. Wells.
\newblock Spectra and {Eigenstates} of {Spin} {Chain} {Hamiltonians}.
\newblock {\em Communications in Mathematical Physics}, 338(1):81--102, August
  2015.

\bibitem[KS20]{kuwahara_gaussian_2020}
Tomotaka Kuwahara and Keiji Saito.
\newblock Gaussian concentration bound and {Ensemble} equivalence in generic
  quantum many-body systems including long-range interactions.
\newblock {\em Annals of Physics}, 421:168278, 2020.

\bibitem[Kuw16]{kuwahara_connecting_2016}
Tomotaka Kuwahara.
\newblock Connecting the probability distributions of different operators and
  generalization of the {Chernoff}–{Hoeffding} inequality.
\newblock {\em Journal of Statistical Mechanics: Theory and Experiment},
  2016(11):113103, November 2016.

\bibitem[KVo17]{kong_spectrum_2017}
Weihao Kong, Gregory Valiant, and {others}.
\newblock Spectrum estimation from samples.
\newblock {\em The Annals of Statistics}, 45(5):2218--2247, 2017.

\bibitem[Mon12]{montanaro_applications_2012}
Ashley Montanaro.
\newblock Some applications of hypercontractive inequalities in quantum
  information theory.
\newblock {\em Journal of Mathematical Physics}, 53(12):122206, December 2012.

\bibitem[Wey12]{weyl_asymptotische_1912}
Hermann Weyl.
\newblock Das asymptotische {Verteilungsgesetz} der {Eigenwerte} linearer
  partieller {Differentialgleichungen} (mit einer {Anwendung} auf die {Theorie}
  der {Hohlraumstrahlung}).
\newblock {\em Mathematische Annalen}, 71(4):441--479, December 1912.

\end{thebibliography}
